\documentclass[10pt,twocolumn]{IEEEtran}
\usepackage{amsmath,amssymb,mathtools,bm,amsthm,graphicx,cite,microtype}
\newtheorem{theorem}{Theorem}

\newtheorem{corollary}{Corollary}
\newtheorem{lemma}{Lemma}
\newcommand{\E}{\mathbb E}
\newcommand{\Pp}{\mathbb P}
\newcommand{\cC}{\mathcal C}
\newcommand{\cS}{\mathcal S}
\newcommand{\A}{\mathcal A}

\usepackage[caption=false,font=small]{subfig}
\usepackage[
    colorlinks=true,
    citecolor=blue,
    linkcolor=red,
    urlcolor=blue
]{hyperref}
\usepackage{orcidlink}

\allowdisplaybreaks[2]

\title{Soft GRAND under Channel Switching and Drift}
\author{Behrooz~Razeghi\textsuperscript{\orcidlink{0000-0001-9568-4166}},
\IEEEmembership{Senior~Member,~IEEE}%
\thanks{B. Razeghi is with Harvard University, Cambridge, MA, USA (e-mail: behroozrazeghi@seas.harvard.edu).}
\thanks{This work was supported by the Swiss National Science Foundation (SNSF) under Grant No.~222339.}}%

\begin{document}
\maketitle

\begin{abstract}
Under channel switching or drift, the posterior used to order Guessing Random Additive Noise Decoding (GRAND) queries can differ from the matched correction
posterior, increasing query rank and finite-budget error. We bound log query rank by matched posterior self-information plus positive log-posterior mismatch;
exact random-subset collision probabilities bound GRAND with abandonment (GRANDAB) error. Under switching, a state-path mixture achieves vanishing error
uniformly over admissible paths below the minimum constituent capacity when log path-class size is sublinear and uniform input is capacity achieving. Under
drift, pilot refresh bounds mismatch and yields a continuous refresh-period minimizer. Generalized-Gaussian BPSK simulations evaluate both regimes.
\end{abstract}


\section{Introduction}

When the channel changes within a codeword or across successive blocks, the posterior used to order queries in soft Guessing Random Additive Noise Decoding (GRAND) can cease to match the current correction posterior. The true noise-effect (correction) pattern can then receive a larger query rank than under matched ordering, which can increase the probability that a finite-budget decoder abandons before reaching it or accepts a competing codeword first. Matched GRAND and soft GRAND construct likelihood- or posterior-ordered queries \cite{duffy2019grand,solomon2020soft}, while Ordered Reliability Bits GRAND (ORBGRAND) constructs an observation-dependent order from reliability ranks \cite{duffy2022ordered}. The mismatch is temporal because channel variation changes the matched correction posterior within a codeword or across blocks, whereas the query-ordering posterior need not coincide with the posterior induced by the current channel.

Prior work addresses channel memory and correlation through Markov-noise ordering \cite{an2022bursts}, Markov-channel models with time-varying
parameters \cite{qin2022pss}, correlation-aware ordering in ORBGRAND-Approximate Independence (ORBGRAND-AI) \cite{duffy2023ai}, exact likelihood-ordered enumeration for correlated Gaussian noise with structured precision matrices \cite{razeghi2026lpgrand}, sequence-reliability ordering for Gaussian intersymbol interference (ISI) channels \cite{li2026isi}, and finite-memory posterior-tail calculations under a specified posterior model \cite{razeghi2026tcgrand}.

Channel-estimation and adaptation methods address a related source of mismatch. Channel-estimation-enhanced GRAND is studied in \cite{millward2024estimation}. Wiame et al.\ evaluate multiple candidate channel estimates to construct soft information for ORBGRAND under fading channel-estimation error \cite{wiame2025channel}. Decoder-provided pilots use decoded codewords to update channel estimates in time-varying channels \cite{bodet2026pilots}. Closed-loop selection between interleaved and noninterleaved transmission, together with adaptation of the GRAND noise model from estimated interference statistics, is considered in \cite{krishnamachari2026closedloop}. Sensing-aided ORBGRAND combines environmental side information with pilot observations to estimate the channel and construct the log-likelihood ratios (LLRs) supplied to the decoder \cite{ge2026sensing}.

Related information-theoretic results address uncertainty and mismatch from different perspectives. Universal noise-guessing decoding is random-coding universal for unknown finite-state additive channels \cite{miyamoto2025finite}, whereas mismatched noise-guessing decoding has been analyzed through error and complexity exponents for additive memoryless channels under mismatched metrics \cite{miyamoto2026mismatch}. Finite-block mismatched guesswork analyzes ordering loss under distribution mismatch \cite{mariona2023mismatch}, while conditional guesswork characterizes rank-tail exponents with side information \cite{girish2026conditional}. Related finite-block results include guesswork-based coding bounds \cite{mariona2024finite}, finite-blocklength analysis of ORBGRAND \cite{li2026finiteorb}, and exact block error rate (BLER), stopping-time distribution, and average number of tests for ORBGRAND-type decoding under a fixed test budget \cite{wan2026performance}. Reliability-rank companding based on the cumulative distribution function (CDF) enables ORBGRAND to achieve the symmetric capacity of general binary-input memoryless channels \cite{li2026cdf}. A complementary channel-uncertainty formulation selects a single observation-dependent correction order over a family of admissible posteriors using minimax posterior-envelope ordering and gives conditions for preservation of the matched ensemble-average error exponent \cite{razeghi2026uncertainty}.

These results address channel memory or correlation, channel estimation and adaptation, static channel uncertainty or mismatched ordering, and
finite-budget decoding behavior. To our knowledge, none provides a temporal posterior-mismatch bound that controls both the rank of the true correction
and finite-budget GRANDAB error under channel switching or drift. We use the positive part of the log-posterior mismatch as this control quantity and
derive separate reliability guarantees for within-block switching and pilot-tracked block-to-block drift.

We study temporal posterior mismatch directly. For the true correction $E$ and channel observation $R$, let $p(\cdot|R)$ denote the matched correction posterior, let $Q(\cdot|R)$ denote the posterior used to order queries, and let $G_Q(E|R)$ denote the rank of $E$ under nonincreasing $Q(\cdot|R)$. Since $G_Q(E|R)Q(E|R)\le1$,
\begin{equation}
\log G_Q(E|R) \le -\log p(E|R) +\left[\log\frac{p(E|R)}{Q(E|R)}\right]^+,
\label{eq:intro-interface}
\end{equation}
where $-\log p(E|R)$ is the matched posterior self-information and the second term is the positive part of the log-posterior mismatch. Thus, the mismatch term bounds the amount by which $\log G_Q(E|R)$ can exceed the matched posterior self-information. Combining \eqref{eq:intro-interface} with the exact finite-population collision probability for the random-subset code ensemble yields a finite-budget GRANDAB error bound.

We obtain two temporal guarantees. First, consider within-block switching among a fixed family $W_1,\ldots,W_K$ of memoryless channels over blocks of length $n$. Let $S_n$ be the maximum allowed number of switches and $\mathcal S_n(S_n)$ the corresponding admissible state-path class. The equal-weight mixture over the path-conditioned posteriors has pointwise log-posterior mismatch at most $\log|\mathcal S_n(S_n)|=\mathcal O(1+S_n\log n)$.
If $S_n\log n=o(n)$, this mismatch penalty is sublinear in $n$, and every coding rate below $\min_k C(W_k)$ is achievable by a deterministic code sequence whose average message-error probability vanishes uniformly over $s^n\in\mathcal S_n(S_n)$, provided uniform input achieves capacity for every $W_k$.

Second, for block-to-block drift over $B$ data blocks, let $V_B$ denote the total parameter variation and let $\varepsilon_m$ bound the uniform expected parameter-estimation error from $m$ pilots. Refreshing the estimate every $q$ blocks gives average tracking error at most $\varepsilon_m+\frac{qV_B}{B}$. Under the posterior-sensitivity condition with constant $L_{\rm post}$, the block-averaged positive log-posterior mismatch, normalized by $n$, is at most $L_{\rm post}\left(\varepsilon_m+\frac{qV_B}{B}\right)$. If $\varepsilon_m\le C_{\rm est}m^{-1/2}$ eventually for some fixed $C_{\rm est}<\infty$ and $\phi=m/(nq)$ is fixed, then, under the continuous relaxation $q>0$ and in the regime where this risk bound applies, the tracking upper bound has the unique minimizer
\begin{equation}
q^\star= \left(\frac{C_{\rm est}}{2v\sqrt{\phi n}}\right)^{2/3},
\qquad v=\frac{V_B}{B}>0.
\label{eq:intro-qstar}
\end{equation}
Small-block generalized-Gaussian (GGN) BPSK simulations evaluate the switching mixture and pilot-refresh rule using exact random-subset GRANDAB error calculations.


\section{Soft GRAND Model and Finite-Block Rank Control}

Let $(\A,\oplus)$ be a finite Abelian group with $|\A|=a$, and use base-$a$ logarithms. Let $\ominus$ denote the corresponding group subtraction, applied componentwise on $\A^n$. In block $b$, let $\cC_{n,b}\subset\A^n$ be a codebook of size $M_n$, and let the channel have density $w_{b,n}(r|x)$ with respect to a common reference measure on the observation space. Fix a deterministic reference decision $h_n(r)\in\A^n$. For each observation $r$, the mapping $x\mapsto e=x\ominus h_n(r)$ is a bijection on $\A^n$; hence the transmitted word induces the true noise-effect (correction) pattern
\[
E_b=X_b\ominus h_n(R_b).
\]

Under the auxiliary prior $X\sim\mathrm{Unif}(\A^n)$, define the correction posterior
\begin{equation}
p_{b,n}(e|r) = \frac{w_{b,n}(r|h_n(r)\oplus e)}
{\sum_{u\in\A^n}w_{b,n}(r|h_n(r)\oplus u)}.
\label{eq:posterior}
\end{equation}
For any fixed codebook and observation, ordering the candidate corrections by nonincreasing $p_{b,n}(\cdot|r)$ is equivalent to ordering the corresponding codewords by nonincreasing channel likelihood $w_{b,n}(r|\cdot)$, because the denominator in \eqref{eq:posterior} is independent of $e$.
More generally, let $Q_n(\cdot|r)$ be any probability mass function on $\A^n$, queried in nonincreasing order with a fixed deterministic tie rule, and let $G_Q(e|r)$ denote the rank of $e$. Then
\begin{equation}
G_Q(e|r)Q_n(e|r)\le1.
\label{eq:rankmass}
\end{equation}
Indeed, each of the first $G_Q(e|r)$ corrections has $Q_n$-mass at least $Q_n(e|r)$.

Now draw $\cC_n$ uniformly from all $M_n$-subsets of $\A^n$ and choose the transmitted word uniformly from $\cC_n$. Conditioned on $X=x$, the remaining $M_n-1$ codewords form a uniform $(M_n-1)$-subset of the other $N-1$ ambient words, where $N=a^n$. Suppose that the true correction has rank $d$ under a codebook-independent query order and that GRAND with abandonment (GRANDAB) stops after $\tau$ queries. Since exactly $d-1$ corrections precede the true correction, the conditional ensemble-averaged decoding-error probability, with abandonment counted as error, is
\begin{equation}
P_{\rm e}^{\rm sub}(d;\tau) =
\begin{cases}
1-\dfrac{\binom{N-d}{M_n-1}}{\binom{N-1}{M_n-1}}, & d\le\tau,\\[1mm]
1, & d>\tau.
\end{cases}
\label{eq:exactcollision}
\end{equation}

For false acceptance, define $s(d,\tau)=\min\{d-1,\tau\}$, the number of queried corrections preceding either the true correction or abandonment. The corresponding conditional ensemble-averaged false-acceptance probability is
\begin{equation}
P_{\rm false}^{\rm sub}(d;\tau)
\! = \! 1   -  \frac{\binom{N-1-s(d,\tau)}{M_n-1}}{\binom{N-1}{M_n-1}}
\! \le \! \min\!\left\{ 1, \! \frac{(M_n-1)s(d,\tau)}{N-1} \right\}.
\label{eq:falseexact}
\end{equation}
In particular, if $d\le\tau$, then $s(d,\tau)=d-1\le\tau-1$. The query order is codebook-independent and can therefore be used with a fixed structured code. However, \eqref{eq:exactcollision}--\eqref{eq:falseexact} are identities for the random-subset ensemble and do not, in general, yield the BLER of that code.

\vspace{15pt}

\section{Finite-Block Reliability under Posterior Mismatch}

The decoder uses a codebook-independent mapping $(r,H_{b-1})\mapsto Q_{b,n}(\cdot|r,H_{b-1})$, where $H_{b-1}$ denotes the side information available before decoding block $b$. Define the matched posterior self-information and positive part of the log-posterior mismatch by
\begin{align}
J_{b,n} &=-\log p_{b,n}(E_b|R_b), \notag\\ 
\Delta_{b,n}^{+} &=
\left[ \log \frac{p_{b,n}(E_b|R_b)}{Q_{b,n}(E_b|R_b,H_{b-1})} \right]^+ .
\label{eq:mismatch}
\end{align}
Let $G_b$ denote the rank assigned to $E_b$ by $Q_{b,n}(\cdot|R_b,H_{b-1})$. By \eqref{eq:rankmass},
\begin{equation}
\log G_b \le J_{b,n}+\Delta_{b,n}^{+}.
\label{eq:rank-mismatch}
\end{equation}

For the random-subset ensemble, each block uses an independent codebook $\cC_{n,b}$ drawn uniformly from the $M_n$-subsets of $\A^n$. The code size and GRANDAB query budget are
\[
M_n=\lfloor a^{nR}\rfloor, \qquad \tau_n=\lfloor a^{n\gamma}\rfloor.
\]
For a fixed codebook $C$, $P_{e,b}^{(n)}(C)$ denotes the decoding-error probability averaged over the uniformly selected message, the channel output, and the side-information randomness. The probabilities and expectations defining $A_n$ and $D_n^+$ below are taken with respect to the corresponding ensemble-induced marginal distribution of $(E_b,R_b,H_{b-1})$.

\begin{theorem}[Finite-Block Reliability under Posterior Mismatch]
\label{thm:dynamic}
Assume $0<\gamma<1-R$. Suppose there exist $h<\gamma$ and $\delta\in(0,\gamma-h)$ such that
\begin{align}
A_n & \coloneqq \frac{1}{B_n} \sum_{b=1}^{B_n}
\Pp\!\left\{ J_{b,n}>n(h+\delta) \right\} \to0,
\label{eq:An}\\
D_n^+ &\coloneqq \frac{1}{B_n n} \sum_{b=1}^{B_n}
\E \Delta_{b,n}^+ \to0.
\label{eq:Dn}
\end{align}
Then
\begin{align}
\frac{1}{B_n} \sum_{b=1}^{B_n} \E_{\cC_{n,b}}
P_{e,b}^{(n)}(\cC_{n,b})
\le A_n &+ \frac{D_n^+}{\gamma-h-\delta} \notag\\
& +
\frac{(M_n-1)(\tau_n-1)}{a^n-1}.
\label{eq:dynbound}
\end{align}
Consequently, the time-averaged ensemble error probability converges to zero.
\end{theorem}

Conditions \eqref{eq:An}--\eqref{eq:Dn} require, respectively, that the
time-averaged probability of matched posterior self-information exceeding
$n(h+\delta)$ vanish and that the block-averaged expected positive part of the
log-posterior mismatch be $o(n)$. The remaining gap
$\gamma-h-\delta>0$ gives the threshold $n(\gamma-h-\delta)$ for the mismatch
term in the abandonment bound.

\begin{proof}
If $G_b>\tau_n$, then, since $G_b$ is integer and $\tau_n=\lfloor a^{n\gamma}\rfloor$, $G_b>a^{n\gamma}$, and hence $\log G_b>n\gamma$. Together with \eqref{eq:rank-mismatch}, this implies that at least one of $J_{b,n}>n(h+\delta)$ or $\Delta_{b,n}^+>n(\gamma-h-\delta)$ must occur. Therefore, by Markov's inequality,
\begin{equation}
\Pp\{G_b>\tau_n\} \le \Pp\{J_{b,n}>n(h+\delta)\}
+ \frac{\E\Delta_{b,n}^+}{n(\gamma-h-\delta)}.
\label{eq:abandon-bound}
\end{equation}
It remains to bound false acceptance on $\{G_b\le\tau_n\}$. Conditional on $H_{b-1}$, $X_b=x$, and $R_b=r$, the query order is fixed independently of $\cC_{n,b}$. Since $\cC_{n,b}$ is independent of the history, the $M_n-1$ nontransmitted codewords form a uniform $(M_n-1)$-subset of the other $a^n-1$ ambient words. Applying \eqref{eq:falseexact} and using $G_b\le\tau_n$ gives
\[
P_{\rm false}
\le
\frac{(M_n-1)(\tau_n-1)}{a^n-1}.
\]
Averaging \eqref{eq:abandon-bound} over $b$ and adding the false-acceptance bound yields \eqref{eq:dynbound}. Since $R+\gamma<1$, its final term converges to zero, completing the proof.
\end{proof}

\begin{corollary}[Reliability with a Reused Codebook]
\label{cor:samecode}
Suppose a single random-subset codebook $\cC_n$ is reused in all blocks and $H_{b-1}$ contains only pilot or channel-side information that is independent of $\cC_n$. Then the bound \eqref{eq:dynbound} holds with $\frac{1}{B_n} \sum_{b=1}^{B_n} \E_{\cC_n} P_{e,b}^{(n)}(\cC_n)$ on the left-hand side. Consequently, whenever the conditions of Theorem~\ref{thm:dynamic} hold, there exists a deterministic sequence of reused codebooks whose time-averaged error probability converges to zero.
\end{corollary}

\section{Within-Block Switching}

Let $W_1,\ldots,W_K$, $K\ge2$, be a fixed finite collection of memoryless channels, each of which has uniform input as a capacity-achieving distribution. Let $p_k(e|r)$ denote the corresponding single-symbol correction posterior defined by \eqref{eq:posterior}. For $0\le S_n\le n-1$, define the admissible state-path class
\begin{equation}
\cS_n(S_n) = \Big\{s^n\in\{1,\ldots,K\}^n : \sum_{i=2}^{n}  \!\mathbf{1}\{s_i\ne s_{i-1}\}\le S_n \Big\}.
\label{eq:pathclass}
\end{equation}
For $s^n\in\cS_n(S_n)$, write $p_{s^n,n}(e^n|r^n) = \prod_{i=1}^{n}p_{s_i}(e_i|r_i)$. Define the equal-weight mixture of the admissible path-conditioned posteriors by
\begin{equation}
Q_n^{\rm sw}(e^n|r^n) = \frac{1}{|\cS_n(S_n)|} \sum_{s^n\in\cS_n(S_n)}
p_{s^n,n}(e^n|r^n).
\label{eq:swmix}
\end{equation}

For constituent channel $W_k$, define $Z_k=-\log p_k(E|R)$, $h_k=\E Z_k$, $\bar h=\max_{1\le k\le K}h_k$. For each fixed channel observation, the mapping between the channel input $X$ and correction $E$ is bijective. Hence, under uniform input, $h_k=H(E|R;W_k)=H(X|R;W_k)$. Because logarithms are base $a$ and uniform input achieves the capacity of $W_k$,
\[
C(W_k)=1-h_k,
\qquad
1-\bar h=\min_{1\le k\le K}C(W_k).
\]

\begin{theorem}[State-Path Posterior Mixture under Switching]
\label{thm:switch}
For every $s^n\in\cS_n(S_n)$ and every $(e^n,r^n)$ such that $p_{s^n,n}(e^n|r^n) > 0$,
\begin{equation}
\log \frac{p_{s^n,n}(e^n|r^n)}{Q_n^{\rm sw}(e^n|r^n)} \le r_n^{\rm sw},
\qquad
r_n^{\rm sw} \coloneqq \log|\cS_n(S_n)|,
\label{eq:swmismatch}
\end{equation}
where
\begin{equation}
|\cS_n(S_n)| = K\sum_{j=0}^{S_n} \binom{n-1}{j}(K-1)^j.
\label{eq:pathcount}
\end{equation}
For fixed $K\ge2$, $r_n^{\rm sw}\le \log K+\log(S_n+1)+S_n\log[n(K-1)]$, and hence $r_n^{\rm sw}=\mathcal O(1+S_n\log n)$. If $S_n\log n=o(n)$, then, for every $\epsilon>0$,
\begin{equation}
\sup_{s^n\in\cS_n(S_n)} \Pp_{s^n}\!\left\{ \frac{1}{n} \log G_{Q^{\rm sw}}(E^n|R^n) > \bar h+\epsilon \right\} \longrightarrow 0.
\label{eq:swrank}
\end{equation}
Moreover, for every $R<1-\bar h = \min_{1\le k\le K}C(W_k)$, there exists a deterministic code sequence whose average message-error probability converges to zero uniformly over $s^n\in\cS_n(S_n)$. Because $\cS_n(S_n)$ contains every constant state path, any achievable rate is at most $\min_{1\le k\le K}C(W_k)$. Together with the preceding achievability result, the capacity under average message error, with reliability required uniformly over deterministic $s^n\in\cS_n(S_n)$, is therefore $\min_{1\le k\le K}C(W_k)$, provided uniform input achieves capacity for every constituent channel.
\end{theorem}

\begin{proof}
Each summand in \eqref{eq:swmix} has weight
$|\cS_n(S_n)|^{-1}$, so
\[
Q_n^{\rm sw}(e^n|r^n) \ge \frac{p_{s^n,n}(e^n|r^n)}{|\cS_n(S_n)|}
\]
for every admissible path, which gives \eqref{eq:swmismatch}. Choosing $j$ switch locations among the $n-1$ possible locations and one of $K-1$ new states at each switch gives \eqref{eq:pathcount}. For fixed $K$, this implies $r_n^{\rm sw}=\mathcal{O}(1+S_n\log n)$. Conditional on $s^n$, the matched posterior self-information is
\[
-\log p_{s^n,n}(E^n|R^n) = \sum_{i=1}^{n}-\log p_{s_i}(E_i|R_i),
\]
a sum of independent random variables with mean at most $n\bar h$. For each $k$ and every $t\in(0,1)$,
\[
\E a^{tZ_k} = \E_R\sum_e p_k(e|R)^{1-t} \le a^t.
\]
Hence the centered cumulant-generating functions are finite in a neighborhood of zero; since $K$ is fixed, a Chernoff bound applies uniformly over all admissible paths. Together with \eqref{eq:rankmass} and \eqref{eq:swmismatch}, this proves \eqref{eq:swrank}.
For $R<1-\bar h$, choose $\epsilon>0$ and $\gamma$ such that $\bar h+\epsilon<\gamma<1-R$. With $\tau_n=\lfloor a^{n\gamma}\rfloor$, the preceding uniform Chernoff bound and the random-subset collision bound give an exponentially vanishing ensemble-average error probability uniformly over the admissible paths. Since $|\cS_n(S_n)|=a^{o(n)}$, the ensemble average of the sum of these error probabilities over all admissible paths also tends to zero. Hence there exists a deterministic code sequence for which the maximum average message-error probability over $\cS_n(S_n)$ tends to zero. Conversely, every constant state path belongs to $\cS_n(S_n)$, so any achievable rate is at most $C(W_k)$ for every $k$. This proves the capacity statement.
\end{proof}

For a fixed candidate correction $e^n$, its switching-mixture score can be evaluated without enumerating all state paths. Let $F_i(k,\ell)$ denote the sum of the path-posterior products over state paths of length $i$ that end in state $k$ and contain exactly $\ell$ switches. Initialize
\[
F_1(k,0)=p_k(e_1|r_1), \qquad F_1(k,\ell)=0 \quad (\ell>0).
\]
For $i\ge2$,
\begin{equation}
\!\!\!
F_i(k,\ell) = p_k(e_i|r_i) \Big[ F_{i-1}(k,\ell) + \sum_{j\ne k}F_{i-1}(j,\ell-1) \Big],
\label{eq:swDP}
\end{equation}
with $F_i(k,\ell)=0$ for $\ell<0$. Therefore
\begin{equation}
Q_n^{\rm sw}(e^n|r^n) = \frac{1}{|\cS_n(S_n)|} \sum_{\ell=0}^{S_n} \sum_{k=1}^{K} F_n(k,\ell).
\label{eq:swDPscore}
\end{equation}
A direct implementation of \eqref{eq:swDP} requires $\mathcal{O}\!\left(nK^2(S_n+1)\right)$ operations per candidate correction; maintaining the sum over the previous states reduces this to $\mathcal{O}\!\left(nK(S_n+1)\right)$. The recursion can be evaluated in the log domain for numerical stability. It is a per-candidate scoring recursion. Naively scoring and sorting all $a^n$ corrections remains exponential in $n$, so scalable enumeration in nonincreasing $Q_n^{\rm sw}$ order requires a separate search procedure.

\section{Pilot Tracking under Block-to-Block Drift}

Fix an arbitrary deterministic parameter path $\theta_1,\ldots,\theta_B\in\Theta$, where $B$ is the number of data blocks, $\Theta\subset\mathbb R^d$ is compact, and $\|\cdot\|$ denotes the Euclidean norm. The set of refresh blocks is
\[
\mathcal T_q = \{1,1+q,1+2q,\ldots\}\cap\{1,\ldots,B\}.
\]
At each $t\in\mathcal T_q$, $m$ pilot observations produce an estimator $\widetilde\theta_t\in\Theta$ satisfying
\begin{equation}
\sup_{\theta\in\Theta} \E_\theta\|\widetilde\theta-\theta\| \le \varepsilon_m .
\label{eq:est}
\end{equation}
For block $b$, let $t(b)=\max\{t\in\mathcal T_q:t\le b\}$ and set $\widehat\theta_b= \widetilde\theta_{t(b)}$. Define the total parameter variation
\begin{equation}
V_B = \sum_{b=2}^{B} \|\theta_b-\theta_{b-1}\|.
\label{eq:VB}
\end{equation}

Assume that the correction posteriors have common support and that, for some $L_{\rm post}<\infty$,
\begin{equation}
\frac{1}{n} \E_\theta \left[ \log \frac{p_{\theta,n}(E^n|R^n)}{p_{\theta',n}(E^n|R^n)} \right]^+
\le L_{\rm post}\|\theta-\theta'\|
\label{eq:sens}
\end{equation}
for all $\theta,\theta'\in\Theta$. The decoder in block $b$ uses $Q_{b,n}=p_{\widehat\theta_b,n}$.

\begin{theorem}[Tracking Bound for Periodic Pilot Refresh]
\label{thm:pilot}
Suppose that, conditional on the parameter path, the pilot observations used to form $\widetilde\theta_t$ are independent of the data observations in the blocks to which that estimate is applied. Then
\begin{align}
\frac{1}{B} \sum_{b=1}^{B} \E\|\theta_b-\widehat\theta_b\|
&\le \varepsilon_m+\frac{qV_B}{B},
\label{eq:track}\\
D_{B,n}^+ \coloneqq \frac{1}{Bn} \sum_{b=1}^{B} \E\Delta_{b,n}^+
&\le L_{\rm post} \left( \varepsilon_m+\frac{qV_B}{B} \right).
\label{eq:pilotreg}
\end{align}
\end{theorem}

\begin{proof}
For $b$ in the refresh interval beginning at $t=t(b)$,
\begin{equation}
\|\theta_b-\widehat\theta_b\|
\le \|\theta_t-\widetilde\theta_t\| + \sum_{\ell=t+1}^{b} \|\theta_\ell-\theta_{\ell-1}\|.
\label{eq:track-pointwise}
\end{equation}
Averaging the first term over the blocks and applying \eqref{eq:est} gives at most $\varepsilon_m$. Each increment $\|\theta_\ell-\theta_{\ell-1}\|$ contributes to at most $q$ blocks, so the second term contributes at most $qV_B/B$. This proves \eqref{eq:track}. Conditional independence of the pilot estimator and the data observations allows \eqref{eq:sens} to be applied conditional on $\widehat\theta_b$; averaging then gives \eqref{eq:pilotreg}.
\end{proof}

\begin{corollary}[Refresh-Period Optimization at Fixed Pilot-to-Data Ratio]
\label{cor:qstar}
Suppose $C_{\rm est}<\infty$ is any fixed constant satisfying $\varepsilon_m\le C_{\rm est}m^{-1/2}$ for all sufficiently large $m$, and define
\[
v=\frac{V_B}{B}, \qquad \phi=\frac{m}{nq}.
\]
Under the continuous relaxation $q>0$, with $m=\phi nq$ in the regime where this risk bound applies, the right-hand side of \eqref{eq:track} is bounded by
\begin{equation}
f(q) = \frac{C_{\rm est}}{\sqrt{\phi nq}} + qv.
\label{eq:fq}
\end{equation}
For $v>0$, $f$ has the unique minimizer
\begin{equation}
q^\star = \left( \frac{C_{\rm est}}{2v\sqrt{\phi n}} \right)^{2/3},
\qquad f(q^\star) = \frac{3C_{\rm est}^{2/3}v^{1/3}}{2^{2/3}(\phi n)^{1/3}}.
\label{eq:qstar}
\end{equation}
For integer-constrained $q$, minimize $f$ over the admissible integers. Since $f$ decreases on $(0,q^\star)$ and increases on $(q^\star,\infty)$ for $v>0$, it suffices to compare the nearest admissible value on each side of $q^\star$. If $v=0$, $f(q)$ decreases with $q$, so the largest admissible refresh interval minimizes the bound. Over a complete $q$-block refresh interval, the fraction of pilot symbols among pilot and data channel uses is $\phi/(1+\phi)$.
\end{corollary}

For the small-block numerical study, we set $\phi=0.5$, giving $m=5q$ and a pilot fraction of $1/3$; over the tested grid, this gives at least $m=20$ pilot observations per refresh.

\begin{figure*}[t]
\centering
\captionsetup[subfloat]{captionskip=-4pt}
\subfloat[\label{fig:dynamic-switching}]{%
    \includegraphics[width=0.4\textwidth]{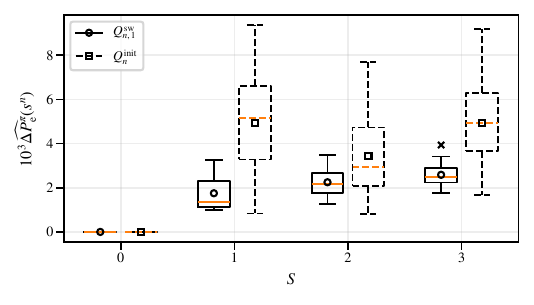}%
}%
\qquad
\subfloat[\label{fig:dynamic-tracking}]{%
    \includegraphics[width=0.4\textwidth]{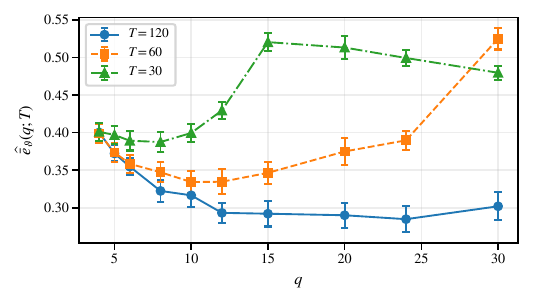}%
}
\\[-4pt]
\subfloat[\label{fig:dynamic-minimizer}]{%
    \includegraphics[width=0.4\textwidth]{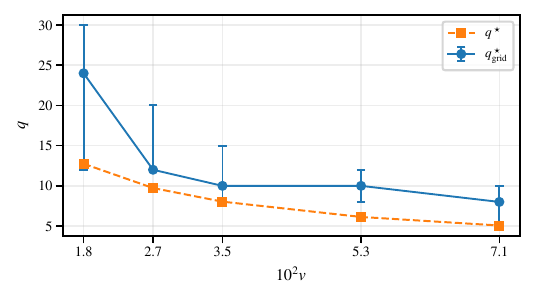}%
}
\qquad
\subfloat[\label{fig:dynamic-decoding}]{%
    \includegraphics[width=0.4\textwidth]{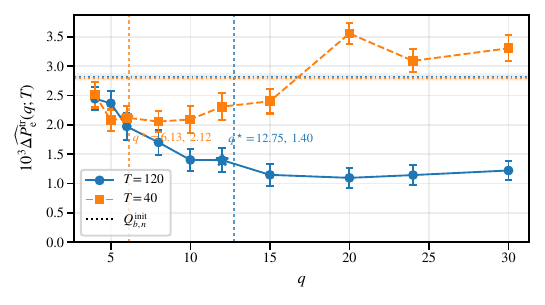}%
}

%
%
\caption{Finite-block numerical evaluation of channel switching and drift.
(a) Pathwise excess GRANDAB error under switching.
(b) Parameter-tracking error versus refresh period.
(c) Empirical grid minimizer and continuous-bound minimizer versus normalized
path variation.
(d) Excess GRANDAB error under pilot-tracked drift.}
\label{fig:dynamic}

\end{figure*}

\subsection{BPSK with Generalized-Gaussian Noise}

Let $\vartheta=(\beta,\eta)$, where $\eta=\ln\sigma$, and consider independent generalized-Gaussian noise samples within each block, each with density
\begin{equation}
f_{\beta,\sigma}(z) = \frac{\beta}{2\sigma\Gamma(1/\beta)}
\exp\!\left[ -\left(\frac{|z|}{\sigma}\right)^\beta \right],
\; 1<\beta<2, \; \sigma>0.
\label{eq:ggn}
\end{equation}
Known BPSK pilots yield residuals $Y-X$ with density \eqref{eq:ggn}. We measure parameter variation in the coordinates $\vartheta=(\beta,\ln \sigma)$ using the norm in \eqref{eq:VB}.

For the sign reference decision and $r=|y|$, the binary correction posterior has natural-log odds
\begin{equation}
\lambda_{\beta,\sigma}(r) \coloneqq
\ln \frac{p_{\beta,\sigma}(E=0|r)}{p_{\beta,\sigma}(E=1|r)}
= \sigma^{-\beta} \left[ (r+1)^\beta-|r-1|^\beta \right].
\label{eq:ggnllr}
\end{equation}

\begin{lemma}[Generalized-Gaussian Posterior Regularity]
\label{lem:ggnreg}
Let $\Theta_0=[\beta_{\min},\beta_{\max}] \times[\eta_{\min},\eta_{\max}] \subset(1,2) \times\mathbb R$, where $1<\beta_{\min}<\beta_{\max}<2$ and $-\infty<\eta_{\min}<\eta_{\max}<\infty$. Then the family of BPSK correction posteriors indexed by $\Theta_0$ has common support, and there exists $L_{\rm post}<\infty$ such that \eqref{eq:sens} holds for all $\vartheta,\vartheta'\in\Theta_0$.
\end{lemma}

\begin{proof}
The density \eqref{eq:ggn} is strictly positive, so both binary correction probabilities are positive for every finite observation. The derivatives of the log posterior with respect to $(\beta,\eta)$ are bounded, uniformly on $\Theta_0$, by an integrable function of the observation of polynomial and log-polynomial growth. Generalized-Gaussian tails on a compact parameter set give finite uniform moments of this envelope. A mean-value bound on the single-symbol log posterior, followed by $[\sum_i x_i]^+\le\sum_i|x_i|$, yields \eqref{eq:sens}.
\end{proof}

For the pilot estimator, assume the uniform parametric risk bound
\begin{equation}
\sup_{\vartheta\in\Theta_0} \E_\vartheta
\|\widehat\vartheta_m-\vartheta\|
\le C_{\rm est}m^{-1/2}
\label{eq:rootm}
\end{equation}
for all sufficiently large $m$. Here $C_{\rm est}$ is any finite eventual upper-bound constant and need not be unique. Our numerical study below uses a constrained maximum-likelihood estimator; the reliability result requires only the risk bound \eqref{eq:rootm}.

Define $Z_\vartheta=-\log p_\vartheta(E|R)$, $h(\vartheta)=\E_\vartheta Z_\vartheta$, $\bar h=\sup_{\vartheta\in\Theta_0}h(\vartheta)$. For every $t\in(0,1)$,
\[
\E_\vartheta 2^{tZ_\vartheta} = \E_R\sum_{e\in\{0,1\}}p_\vartheta(e|R)^{1-t}
\le 2^t
\]
uniformly over $\vartheta\in\Theta_0$. Hence $\sup_{\vartheta\in\Theta_0}\E_\vartheta Z_\vartheta^2<\infty$. Since the symbols are independent within each block and $h(\vartheta)\le\bar h$, Chebyshev's inequality gives, for every $\epsilon>0$,
\[
\sup_{\vartheta\in\Theta_0}
\Pp_\vartheta\!\left\{ \frac{1}{n}\sum_{i=1}^{n}Z_{\vartheta,i}>\bar h+\epsilon \right\}
\longrightarrow0,
\]
which is the uniform upper-tail concentration required in Theorem~\ref{thm:dynamic}.

\begin{corollary}[Reliability under Generalized-Gaussian Drift]
\label{cor:ggrel}
Let $B_n$, $m_n$, $q_n$, and $V_{B_n}$ vary with blocklength $n$, and suppose the parameter path remains in $\Theta_0$. If $R<1-\bar h$, choose $\gamma$ such that $\bar h<\gamma<1-R$. With $Q_{b,n}=p_{\widehat\vartheta_b,n}$, the time-averaged error probability in Theorem~\ref{thm:dynamic} converges to zero provided
\[
m_n\to\infty, \qquad \frac{q_nV_{B_n}}{B_n}\to0.
\]
The rate $R$ is measured per data symbol. If pilot and data symbols consume the same channel resource, asymptotically vanishing pilot overhead further requires
\[
\frac{m_n}{nq_n}\to0.
\]
For example, $B_n=n$, $q_n=\lceil n^{1/4}\rceil$, $m_n=\lceil n^{1/2}\rceil$, $V_{B_n}=\mathcal O(1)$ satisfy all three limits.
\end{corollary}

\section{Finite-Block Numerical Evaluation}

Figure~\ref{fig:dynamic}(a) uses two BPSK-GGN channel states, $W_1$ with $(\beta,\sigma)=(1.05,0.75)$ and $W_2$ with $(1.95,1.30)$. For this pair, $\bar h=0.46709$ bits/symbol. We set $n=10$, $(R,\gamma)=(0.25,0.65)$, $M_n=5$, and $\tau_n=90$, so that $\bar h<\gamma<1-R$. The oracle uses $Q_n^{\rm oracle}=p_{s^n,n}$, whereas the static baseline uses
\[
Q_n^{\rm init}(e^n|r^n) = \prod_{i=1}^{n}p_1(e_i|r_i).
\]
Both non-oracle query distributions use the known initial state $s_1=1$. For $S\in\{0,1,2,3\}$, we evaluate all $1$, $9$, $36$, and $84$ paths, respectively, that start in state~1 and have exactly $S$ switches. For each $S$, $Q_n^{\rm sw}$ is the uniform mixture over paths starting in state~1 with at most $S$ switches.

For $\pi\in\{\mathrm{sw},\mathrm{init}\}$, define
\[
\Delta P_{\mathrm e}^{\pi}(s^n) \coloneqq
P_{\mathrm e}^{\pi}(s^n) - P_{\mathrm e}^{\mathrm{oracle}}(s^n).
\]
Figure~\ref{fig:dynamic}(a) shows Tukey boxplots, across paths, of the Monte Carlo estimates of $10^3\Delta P_{\mathrm e}^{\pi}(s^n)$ at each exact switch count. We use 64 channel observations per path and enumerate all $2^{10}$ corrections. At $S=3$, the means across the 84 paths are $2.586\times10^{-3}$ for $Q_n^{\rm sw}$ and $4.938\times10^{-3}$ for $Q_n^{\rm init}$. The estimated pathwise excess error is smaller under $Q_n^{\rm sw}$ on 81 of the 84 paths. As an implementation check, we verify the fixed-initial-state form of \eqref{eq:swmismatch} for every sampled observation and every correction. The state transition changes both $\beta$ and $\sigma$ and therefore is not a one-parameter $E_b/N_0$ sweep.

For the drift experiments, $B=120$ and
\[
\bar e_{\vartheta}(q) \coloneqq
\frac{1}{B}\sum_{b=1}^{B} \E\|\vartheta_b-\widehat\vartheta_b\|.
\]
We use the deterministic parameter path
\begin{align*}
\beta_b &=1.5+0.42\sin(2\pi b/T-\pi/2),\\
\eta_b &=0.28\sin(2\pi b/T-\pi/2+0.7).
\end{align*}
Across the tested values of $T$, the maximum of $h(\vartheta_b)$ along these paths is at most $0.52262$ bits/symbol. We use the same $(R,\gamma)=(0.25,0.65)$ and set
\[
\mathcal Q=\{4,5,6,8,10,12,15,20,24,30\}.
\]
We set $\phi=0.5$, giving $m=5q$ and a pilot fraction of $1/3$. Thus, each refresh uses at least $m=20$ pilot observations over the tested grid. The static baseline is $Q_{b,n}^{\rm init}=p_{\vartheta_1,n}$.

Let $\widehat{\bar e}_{\vartheta}(q)$ denote the sample estimate of $\bar e_{\vartheta}(q)$. Figure~\ref{fig:dynamic}(b) shows $\widehat{\bar e}_{\vartheta}(q)$ for $T\in\{120,60,30\}$. The bars are pointwise 95\% Student-$t$ confidence intervals based on 64 independent pilot-estimation repetitions.

For numerical comparison with \eqref{eq:qstar}, we use the finite-grid calibration $\widehat C_{\rm cal}=3.599$. This value is the maximum of the upper endpoints of 50 separate pointwise 95\% Student-$t$ confidence intervals for $\E[\sqrt m\,\|\widehat\vartheta_m-\vartheta\|]$, obtained from five parameter values, ten pilot sizes, and 64 repetitions per setting.

For each $T\in\{120,80,60,40,30\}$, $\widehat{\bar e}_{\vartheta}(q)$ has a unique minimizer over
$\mathcal Q$,
\[
\widehat q_{\rm grid}^{\star} = \arg\min_{q\in\mathcal Q} \widehat{\bar e}_{\vartheta}(q).
\]
With $v=V_B/B$, define the calibrated continuous surrogate
\[
\widehat f_{\rm cal}(q) \coloneqq \frac{\widehat C_{\rm cal}}{\sqrt{\phi nq}}+qv,
\]
whose minimizer is
\[
\widehat q_{\rm cal}^{\star} = \left(\frac{\widehat C_{\rm cal}}{2v\sqrt{\phi n}}\right)^{2/3}.
\]
Figure~\ref{fig:dynamic}(c) compares $\widehat q_{\rm grid}^{\star}$ and $\widehat q_{\rm cal}^{\star}$ as functions of $v$. The intervals for $\widehat q_{\rm grid}^{\star}$ are central 95\% bootstrap intervals based on 5000 resamples. As $T$ decreases from $120$ to $30$, $v$ increases from $0.01768$ to $0.07058$, $\widehat q_{\rm grid}^{\star}$ decreases from $24$ to $8$, and $\widehat q_{\rm cal}^{\star}$ decreases from $12.75$ to $5.07$. The latter minimizes $\widehat f_{\rm cal}$ rather than the sampled finite-block tracking or decoding error.

For the decoding experiment, define
\[
\Delta P_{\mathrm e}^{\mathrm{tr}}(q) \coloneqq
P_{\mathrm e}^{\mathrm{tr}}(q) - P_{\mathrm e}^{\mathrm{oracle}}.
\]
Figure~\ref{fig:dynamic}(d) shows Monte Carlo estimates of $10^3\Delta P_{\mathrm e}^{\mathrm{tr}}(q)$ for $T\in\{120,40\}$. For $T=120$ and $40$, $\widehat q_{\rm cal}^{\star}$ equals $12.75$ and $6.13$, respectively. The stars mark the nearest tested values $q=12$ and $6$. The corresponding estimated excess errors are $1.399\times10^{-3}$ and $2.122\times10^{-3}$. The empirical minima are $1.097\times10^{-3}$ at $q=20$ and $2.056\times10^{-3}$ at $q=8$, respectively. The dotted lines show the static baseline $Q_{b,n}^{\rm init}$, whose estimated excess errors are $2.821\times10^{-3}$ and $2.797\times10^{-3}$, respectively. The bars are 95\% Student-$t$ confidence intervals for the paired differences, based on 64 repetitions.

For Figs.~\ref{fig:dynamic}(a) and \ref{fig:dynamic}(d), all $2^n$ corrections are enumerated. For each sampled observation, \eqref{eq:exactcollision} is averaged exactly under the matched correction posterior. The random-subset codebook average is therefore evaluated analytically. The channel and pilot averages are estimated by Monte Carlo, while the random-subset codebook average is evaluated analytically without simulating random codebooks.

\section{Conclusion}

Channel variation can make the posterior used to order soft GRAND queries mismatched to the current channel even when the codebook is fixed. The positive part of the log-posterior mismatch bounds the amount by which log query rank can exceed matched posterior self-information; with exact random-subset collision probabilities, this yields finite-budget GRANDAB error bounds. For within-block switching, the state-path mixture has sublinear log-posterior mismatch when $S_n\log n=o(n)$. For block-to-block drift, periodic pilot refresh bounds the average positive part of the log-posterior mismatch; under root-$m$ estimation error and nonzero drift, the tracking upper bound has an explicit refresh-period minimizer.

\bibliographystyle{IEEEtran}
\bibliography{refs}

\end{document}